\documentclass{article}

\usepackage[preprint,nonatbib]{nips_2018}

\usepackage[utf8]{inputenc} 
\usepackage[T1]{fontenc}    
\usepackage{hyperref}       
\usepackage{url}            
\usepackage{booktabs}       
\usepackage{wrapfig}
\usepackage{lipsum}
\usepackage{amsfonts}       
\usepackage{nicefrac}
\usepackage[title]{appendix}
\usepackage{wrapfig}
\usepackage{microtype}      
\usepackage{fancyhdr}
\usepackage{algorithm}
\usepackage[noend]{algpseudocode}
\usepackage{comment}
\usepackage{times}
\usepackage{tikz}
\usepackage{amsmath}
\usepackage{multirow}
\usepackage{multirow}
\usepackage{changepage}
\usepackage{amssymb}
\usepackage{graphicx}%
\newtheorem{theorem}{Theorem}

\newtheorem{claim}{Claim}

\newtheorem{corollary}{Corollary}

\newtheorem{lemma}{Lemma}

\newenvironment{proof}[1][Proof]{\textbf{#1.} }{\ \rule{0.5em}{0.5em}}
\PassOptionsToPackage{options}{natbib}
\title{Straggler-Resilient and Communication-Efficient Distributed Iterative Linear Solver}

\author{
  Farzin Haddadpour \\
  Pennsylvania State University\\
  \texttt{fxh18@psu.edu}
   \And
  Yaoqing Yang\\
  Carnegie Mellon University \\
  \texttt{yyaoqing@andrew.cmu.ed} \\
     \And
 {Malhar Chaudhari}\\
  Carnegie Mellon University \\
    \texttt{mschaudh@andrew.cmu.edu} \\
       \And
      {Viveck R Cadambe}\\
      Pennsylvania State University \\
        \texttt{viveck@engr.psu.edu} \\
   \And
  {Pulkit Grover} \\
   Carnegie Mellon University \\
  \texttt{pgrover@andrew.cmu.edu} \\
}

\begin{document}

\maketitle

\begin{abstract}
We propose a novel distributed iterative linear inverse solver method. Our method, {\emph{PolyLin}}, has significantly lower communication cost, both in terms of number of rounds as well as number of bits, in comparison with the state of the
art at the cost of higher computational complexity and storage. Our algorithm also has a built-in resilience to straggling and faulty computation nodes. We develop a natural variant of our main algorithm that trades off communication cost for computational complexity. Our method is inspired by ideas in error correcting codes.
\end{abstract}

\section{Introduction}
Over the last decade, owing to the increasing {data} volumes, data processing is commonly done in large distributed systems.  The collective processing capacity of multiple computing nodes operating in parallel is higher, and as a consequence, distributed algorithms are significantly faster and often more memory efficient, as compared to single node computations. However, the performance of such distributed algorithms {does not scale linearly with the number of nodes (beyond a few tens of nodes)} in practice because of two major bottlenecks \cite{qi2016paleo, dean2012large}. First, the overhead of communicating data becomes significant in comparison to the actual computation time. Second, a few excessively slow nodes called \emph{stragglers} often slow down the overall computation. The {goal} of our paper is to develop new distributed algorithms for linear inverse solvers that overcome communication bottlenecks and stragglers.

Our specific focus is on iterative methods for linear inverse problems, where each iteration has the form $\mathbf{x}_{\textrm{new}} = \mathbf{A}\mathbf{x}_{\textrm{old}}+\mathbf{Q}\mathbf{y},$ where $\mathbf{y}$ is the input to the system. This method includes, as its special case, a wide array methods such as the Jacobi iterative method, Gauss Siedel method, power iterations and Pagerank \cite{atkinson2008introduction,heath2002scientific,kang2009pegasus,gleich2004fast}. Such iterative techniques are used extensively for solving linear inverse problems arising in imaging and inference, both for dense and sparse problems, because of their low complexity \cite{zhang2009cuda, atkinson2008introduction, bertero1998introduction}. A naive baseline method to implement such algorithms in a distributed network of processing nodes is to distribute $\mathbf{A}$ among the nodes, and perform the matrix-vector multiplication $\mathbf{A}\mathbf{x}_{\textrm{old}}$ in distributed manner. When performed over a set of $P$ distributed worker nodes, the per-node computational and storage complexities are both $\frac{1}{P}$ of the overall computational and storage complexities of the centralized algorithm. The main {difficulty with such} distributed implementations, however, is that each node only has a local view of vector $\mathbf{x}_{new}$, and messages need to be exchanged in \emph{every iteration} to concatenate/aggregate the outputs of all nodes before proceeding to the next iteration.

\subsection*{Summary of Contributions}
Our main contribution is a new, non-trivial method of parallelizing iterative linear inverse solvers that requires significantly smaller communication overhead as compared with state of the art, both in terms of the number of rounds of communication and the number of bits communicated.  In addition to having smaller communication,  our algorithms have built in resilience to straggling and/or faulty computing nodes. We present our algorithms and results in a master-worker architecture, where there are one master node and $P$ worker nodes. The worker nodes, which operate in parallel, carry the computational burden of the iterative linear inverse solver.

The master node performs any necessary pre-processing of the inputs and sends them to the worker nodes. It also accumulates results of worker nodes, performs any necessary post-processing, and sends the results back to the worker nodes for the next set of iterations if necessary. For instance, in the baseline algorithm explained above, to perform $n$ iterations, $n$ rounds of communication are performed between the master and the worker node.

In Section \ref{sec:parallel-extreme1}, we develop the \emph{PolyLin} algorithm which has only \emph{one} round of communication between the master node and the worker nodes, irrespective of the number of iterations performed. {PolyLin guarantees linear convergence, that is, the error decays exponentially in the number of iterations.} The core idea of PolyLin is the development of a technique that computes the $n$-th power of $\mathbf{A}$ without requiring $n$ rounds of communication, or requiring every node to store $\mathbf{A}$ entirely. The total amount of data exchanged between the master node and each worker node is approximately twice the length of the target vector $\mathbf{x}$. In contrast, to perform $n$ iterations, the baseline scheme requires $n$ rounds of communication, and the amount of data exchanged is $2n$ times the length of the target vector $\mathbf{x}$, i.e., it scales linearly in the number of iterations. {Additionally, our algorithm has the property that it tolerates stragglers or faulty nodes, that is, in a system with $P$ processing nodes, it suffices if any $K$ of the $P$ nodes complete their job per round of communication; the per-node processing complexity and memory depends on $K$.}

Our algorithm incurs two penalties as compared to the baseline algorithm. First, {PolyLin} requires a fraction of {$\frac{1}{K^{2/n}}$} of the overall computational and storage complexity of the centralized algorithm {where $K\leq P$}; note that the baseline algorithm requires a fraction of $\frac{1}{P}$. Second, {PolyLin} incurs a pre-processing computation cost that, in effect, is equivalent to running the centralized algorithm once in each worker node. {Our experiments in Section \ref{experiment}, reveal that, despite the computational complexity overhead of our algorithms, the communication overhead reduction translates to faster completion time.}  The pre-processing depends only on linear system matrices/vectors $\mathbf{A},\mathbf{Q}$, and therefore can be amortized over multiple uses of the algorithm corresponding to different instances of the input $\mathbf{y}$. From an application viewpoint, the pre-processing cost may be worth the gains in communication and computation complexities when the matrices involved $\mathbf{A},\mathbf{Q}$ are fixed or slowly changing with respect to the target vector, e.g., inference and web queries. One specific example is the personalized PageRank problem where web queries may focus on different topics \cite{haveliwala2003topic} but the underlying linear system matrix (the graph adjacency matrix) does not change. Some other examples include supervised $\ell_2$-minimization with multi-instance learning and multi-label learning and solving numerical equations, such as Poisson equations, with the same system matrix but different inputs.

We develop a variant of {PolyLin} in Section \ref{sec:tradeoff} that trades off the number of rounds with respect to the computation and storage complexities. For instance, by choosing parameters correctly, Section \ref{sec:tradeoff} can be used to develop an algorithm that halves the communication cost as compared to the baseline distributed implementation but incurs a computational and storage complexity penalty of a factor of $\sqrt{P}$. {PolyLin} works by storing carefully constructed linear projections of the matrix $\mathbf{A}$ at each worker as a part of pre-processing. The worker nodes run an iterative algorithm based on the stored linear projections and the initialized target vector $\mathbf{x}$, and send the output to the master node. The coefficients for the linear projections are chosen to be evaluations of certain polynomials to ensure that the master node, on receiving the outputs of the worker node, can recover the solution of the iterative linear solver using polynomial interpolation. Robustness is built into the algorithm in a manner that is similar in spirit to Reed Solomon codes \cite{roth2006introduction}, where, the number of polynomial evaluations chosen is higher than the degree of the polynomial so that a few slow or faulty nodes can be ignored in the interpolation.

\subsection*{{Related Works}}

{Our} work relates to a long line of work that we review in three categories: distributed optimization, distributed linear system solvers and coding theoretic ideas for straggler resilient distributed computing.

 \textbf{1) Distributed optimization}: One can consider the linear inverse solver to be equivalent to a linear regression problem of the form \begin{align}\arg\min_x\|\mathbf{M}\mathbf{x}-\mathbf{y}\|^2_2\label{eq:optimization}.\end{align} Therefore, we compare our results with other communication-efficient distributed optimization algorithms\footnote{{As we mostly focus on square linear systems, we compare with distributed optimization in a regime where the number of data points is comparable to the number of features.}}. We have the following two cases:

{\textbf{1.1) One-shot communication schemes:} At one extreme, there are distributed methods that require only a single round communication such as \cite{mcdonald2009efficient,zinkevich2010parallelized,zhang2012communication,mcwilliams2014loco,heinze2016dual}. In these works, the data is distributed to the worker nodes, and {each worker node solves a ``local'' optimization on the part of the input stored at the worker, and the master node averages the results of the workers. As a consequence, the convergence of these one-shot algorithms is not linear; in fact, some of these algorithms can not guarantee convergence rates beyond what could be achieved if we ignore data residing on all but a single computer \cite{shamir2014communication}.} {In contrast, our one-round algorithm achieves linear convergence although with some computational and storage overhead.}

{{\textbf{1.2) Multiple-round communication schemes:}}  {In order to compare our schemes with algorithms in \cite{lee2017distributed,zhang2015disco,shamir2014communication,ma2015adding,nesterov2013introductory}, we focus on the regime where the number of data points is comparable to the number of features.} {We show that for our algorithms, when the number of iterations $n$ satisfies $n\geq \log_{\frac{1}{|\sigma_1|}}\frac{N\underset{1\leq i\leq N}{\max}|\alpha_i|}{\epsilon}$, the error is upper bounded by $\epsilon$, where $N$ is the number of data points, $\sigma_1$ is the second largest eigenvalue of $\mathbf{A}$, and $\alpha_i$ is the projection of an arbitrarily chosen initialization vector on to the $i$-th eigen vector of $\mathbf{A}$.} \cite{lee2017distributed} introduces two algorithms {\emph{distributed stochastic variance reduced gradient} (DSVRG) and \emph{distributed accelerated stochastic variance reduced gradient} (DASVRG)}. These algorithms respectively require $(1+\frac{\kappa P}{N})\log(\frac{1}{\epsilon})$ and $(1+\sqrt{\frac{\kappa P}{N}})\log(1+{\frac{\kappa P}{N}})\log(\frac{1}{\epsilon})$ rounds of communication to find the optimal solution with error $\epsilon$, where  $\kappa$ is the condition number as defined in \cite{lee2017distributed}. The corresponding computational and communication costs are $O(\frac{N^2}{P}+N\kappa)\log(\frac{1}{\epsilon})$ and $O(\frac{N^2}{P}+N\sqrt{\frac{\kappa P}{N}})\log(1+\frac{\kappa P}{N})\log(\frac{1}{\epsilon})$ respectively. Note that for the case that {$\kappa=\Omega({\frac{N}{P}})$}, our baseline algorithm outperforms DSVRG and DASVRG in terms of computational and communication cost. { While \cite{lee2017distributed} provides a lower bound of $\Omega(\sqrt{\frac{P\kappa}{N}})$ on the number of rounds of communication for a certain class of algorithms where the algorithm of our paper  storage cost is fixed to be that of the baseline algorithm\footnote{To the best of our understanding, the lower bound of \cite{lee2017distributed} requires a specific method of sampling and storing the data as well.} (in an order sense), in our paper we present an algorithm with a fewer number rounds although with a higher storage cost.} For instance, by loading higher computation task at each processing node, we can achieve linear convergence using only one round of communication.

 {Additionally, there have been several other algorithms such as Disco \cite{zhang2015disco}, Dane \cite{shamir2014communication}, COCOA$^{+}$ \cite{ma2015adding} and accelerated gradient method \cite{nesterov2013introductory}.} We compare algorithms {in \cite{lee2017distributed,zhang2015disco,shamir2014communication,ma2015adding,nesterov2013introductory,karakus2017straggler,maity2018robust}} in terms of communication and computational costs in Table \ref{table:comp} {briefly} and in more detail in the appendix. Specifically, our results imply that when $N$ is of the same scaling as $1/\epsilon$ and the number of features, our algorithm outperforms the baseline while the baseline is comparable with these algorithms in terms of communication rounds.}

{\textbf{2) Distributed linear inverse solver using network properties:}} {There is a second line of related works in \cite{demmel2008avoiding,demmel2012communication,wicky2017communication,lipshitz2012communication}
which minimizes the communication cost for various network architectures. {In contrast,} our work does not have these specific network structures. Some of these methods build upon the Krylov-subspace methods. In this work, we mainly focus on stationary methods (such as Jacobi and power iterations) which have successive matrix-vector multiplications. For the specific problem of PageRank, the convergence of Krylov methods strongly depends on the properties of the graph and is non-monotonic. Although Krylov methods have gained popularity, the techniques presented in our paper are still relevant for many specific problems and systems where power iterations perform comparably, or better than Krylov subspace method. Power-iteration and Jacobi methods have approximately the same {convergence rate} determined by the \emph{teleport} probability and the most stable convergence pattern \cite{gleich2004fast}. This is an advantage of stationary methods that perform the same amount of work per any iteration.} Extending our ideas to perform multiple iterations of non-stationary algorithms with fewer rounds of communication, without having to distribute all the data to all the nodes, is an interesting area of future work.

 {\textbf{3) Coding Theoretic approaches for straggler-resilient distributed computing:}} Our algorithms are related to and inspired by recent work that uses error correcting codes for protecting distributed linear operations and optimization problems \cite{dutta2016short, MatDot,MatDot_Allerton, raviv2017gradient, tandon2017gradient,raviv2017gradient, yang2017coding,lee2017speeding, mallick2018rateless,charles2017approximate,karakus2017straggler,ye2018communication,yu2017polynomial,yu2018straggler} from faults and stragglers.
{Specifically, the approaches of  \cite{karakus2017straggler, tandon2017gradient, raviv2017gradient, yang2017coding, charles2017approximate, karakus2017straggler, ye2018communication} {can be} interpreted in our context as the introduction of coding methods for adding straggler/fault tolerance to the baseline algorithm. Since these papers essentially include a variant of the baseline algorithm, the communication overheads, measured in terms of number of rounds as well as the number bits, are proportional to the number iterations of the power method. The main contribution of our work to this body of literature is to develop a novel method that can perform multiple iterations of the power method in a single communication round, thereby reducing the overall communication cost (in addition to providing straggler resilience).}

From a technical viewpoint, our core ideas are related to references \cite{yu2017polynomial, MatDot, MatDot_Allerton, yu2018straggler} which use polynomial evaluation based error correcting codes to protect matrix multiplications from faults, stragglers and errors. Our approach particularly builds on \cite{MatDot} which multiplies multiple (more than two) matrices in a straggler resilient manner. The process of adapting the ideas of \cite{MatDot} to the power method however requires the development of new ancilliary results; the relevent ideas of \cite{MatDot} as well as some new related results are described in Sec. \ref{sec:Propertiesofmatrixpoly}.


\section{Background}
In this section, we provide some preliminary background on linear inverse solvers.

\textbf{Preliminaries on Solving Linear Systems using Iterative Methods:}
Consider the linear inverse problem of finding {an} $N\times 1$ vector $\mathbf{x}$
that satisfies $\mathbf{M}\mathbf{x}=\mathbf{y}$, given a rank  $m$, $L\times N$ matrix $\mathbf{M}$ and a $L\times 1$ vector  $\mathbf{y}$.
 When $\mathbf{M}$ is a square full rank matrix, the closed-form solution is $\mathbf{x}=\mathbf{M}^{-1}\mathbf{y}$.
Two iterative methods, namely the Jacobi and the gradient descent method are used to solve these linear inverse problems: \emph{ Jacobian Method for Square System}
  and \emph{Gradient Descent Method} (see \cite{saad2003iterative} and Appendix \ref{a}).

We can cast both iterative methods into the same formulation as
 \begin{align}
 \mathbf{x}^{(n+1)}=\mathbf{A}\mathbf{x}^{(n)}+\mathbf{Q}\mathbf{y}\label{eq:General-iterative}
 \end{align}
 for two appropriate matrices $\mathbf{A}$ and $\mathbf{Q}$. Denote by $\mathbf{x}^*$ the
 fixed point of  (\ref{eq:General-iterative}), we have
 $\mathbf{x^{*}}=\mathbf{A}\mathbf{x^{*}}+\mathbf{Q}\mathbf{y}$.
 Then, using (\ref{eq:General-iterative}) and defining $\mathbf{e}^{(n)}=\mathbf{x}^{(n)}-\mathbf{x}^{*}$, we have $
 \mathbf{e}^{(n)}=\mathbf{A}\mathbf{e}^{(n-1)}=\mathbf{A}^n\mathbf{e}^{(0)}$. Throughout this paper, we assume the absolute values of the eigenvalues of matrix $\mathbf{A}$ are strictly less than $1$, $\lim_{n\longrightarrow\infty}\mathbf{A}^{n}=0$,  which further implies $\lim_{n\longrightarrow\infty}\mathbf{e}^{(n)}=0$. {This condition} guarantees convergence of the iterative method (\ref{eq:General-iterative}).

\textbf{Bound on error:} We mention a lower bound on the error of the iterative method  (\ref{eq:General-iterative}) as a function of the number of iterations $n$. For simplicity, we assume $\mathbf{A}$ is diagonalizable and full rank\footnote{The case when $\mathbf{A}$ is not diagonalizable can also be analyzed with the Jordan decomposition.}.
\begin{lemma}[\textbf{Bound on error}]\label{lem:nupbound}
If the absolute values of the eigenvalues of $\mathbf{A}$ are strictly less than 1 and the number of iterations satisfies $n\geq \log_{\frac{1}{|\sigma_1|}}\frac{N\underset{1\leq i\leq N}{\max}|\alpha_i|}{\epsilon}$, then $\|\mathbf{e}^{(n)}\|\leq \epsilon$.
\end{lemma}
The proof can be found in Appendix \ref{b}. In the sequel, the following equation, which is a consequence of (\ref{eq:General-iterative}), will be useful:
\begin{align}
 \mathbf{x}^{(n)}=\mathbf{A}^{n}\mathbf{x}^{(0)}+(\mathbf{A}^{n-1}+\dots+\mathbf{I})\mathbf{Q}\mathbf{y}\label{eq:equivaldesc}.
\end{align}

Our goal is to implement linear inverse solvers, i.e., solutions to (\ref{eq:General-iterative}) via (\ref{eq:equivaldesc}) in a distributed manner.

\textbf{Notation}: Throughout this paper, we assume that $\mathbf{A}$, $\mathbf{Q}$ and the identity matrix $\mathbf{I}$ are  $N\times N$ matrices, $n$ denotes the number of iterations, and associated error is $\| e^{(n)}\|$, where we denote $\ell_2$ norm of vector $\mathbf{v}$ by $\|\mathbf{v}\|$. Because preserving order of matrices in multiplication is important, we use $\Pi_{i= i_0}^{i_1} \mathbf{M}_i$ to denote $\mathbf{M}_{i_0}\mathbf{M}_{i_0+1}\dots\mathbf{M}_{i_1-1}\mathbf{M}_{i_1} $ when $i_0 < i_1$ and $\mathbf{M}_{i_0} \mathbf{M}_{i_0-1} \dots \mathbf{M}_{i_1+1}\mathbf{M}_{i_1}$ when $i_0 > i_1$.\vspace{-0.5 em}
\section{Setup and BaselineParallel Algorithm}\label{sec:Problemstatement}
\textbf{Setup:} Our setup consists of a master node and $P$ distributed worker/processing nodes. As a part of our algorithm's offline computations, the master node receives as input the matrices $\mathbf{A}, \mathbf{Q} \in \mathbb{R}^{N\times N}$. It does some pre-processing on these inputs and sends some matrices to the worker nodes, which store the received inputs. During online computations, the master node receives a vector $\mathbf{y}$ as input and outputs $\mathbf{x}^{(n)}$ as per (\ref{eq:General-iterative}), through an algorithm that interacts with the worker nodes. We assume that $\mathbf{x}^{(0)}$ is initialized arbitrarily.

We consider algorithms that operate in \emph{rounds}, where each round consists of a communication from the master node to the worker nodes, and communication from the worker nodes back to the master node. In Sections \ref{sec:parallel-extreme1} and \ref{sec:tradeoff}, we focus on algorithms where the master node waits for the fastest $K$ workers to finish before proceeding to the next round. A round can possibly correspond to multiple iterations of (\ref{eq:General-iterative}). We measure the performance of our algorithms as follows:

\emph{Communication cost}: We use a linear model \cite{bruck1997efficient,bertsekas1989parallel,bang2008digraphs} to measure the communication cost. The cost of a round of communication that involves sending of $B_1$ bits from the master node to every worker node, and $B_2$ bits from every worker node to the master node is measured as $\beta_1 + \beta_2(B_1 + B_2)$. Thus in this model, the coefficient against $\beta_1$ represents the number of rounds used by the algorithm, and the coefficient against $\beta_2$ represents the number of bits exchanged between a master node and each worker node.

We compare various algorithms in terms of the \emph{complexity of workers} done as a part of online computations and their \emph{storage cost}. We present \emph{pre-processing} complexity due to offline computations performed by the master node before the online computations, as well as the \emph{post-processing} complexity of the master node after gathering the outputs of the workers. For an algorithm where in each round, the master node waits for the fastest $K$ workers, the \emph{straggler tolerance} is measured as $(P-K)$. 

We study the dependence of the above cost metrics in terms of the parameters-$N$, $P$, and $n$, which implicitly reflects the error rate as per Lemma \ref{lem:nupbound}. All the costs are indicated in Table \ref{table:comp} and the calculations can be found in the Appendix \ref{d}.

\begin{figure}
  \begin{center}
    \includegraphics[scale=0.55]{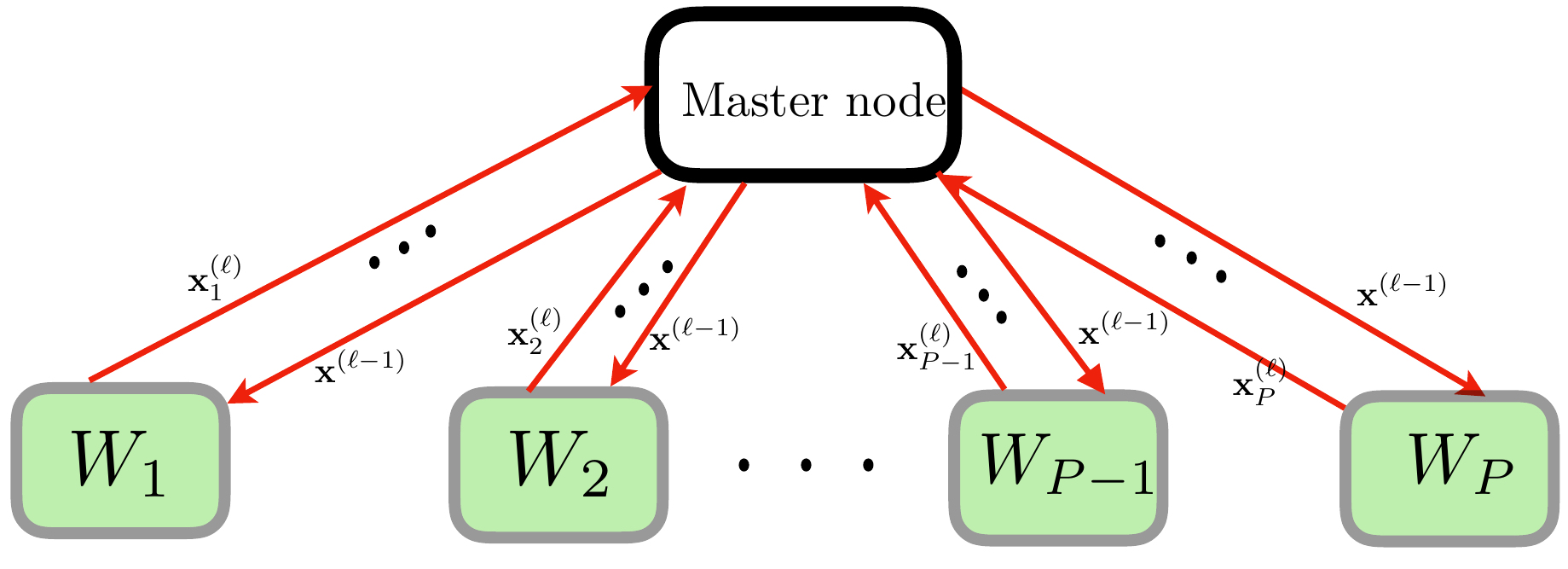}
  \end{center}
\vspace{-0.5 em}
\caption{In BaselineParallel algorithm each worker node communicates independently of others with the master node and the master node waits all of the worker node to finish their task at each communication round.} \label{fig:0}\vspace{-0.5 em}
\end{figure}

\textbf{BaselineParallel:} First, consider a centralized implementation of a linear inverse solver where a single node implements (\ref{eq:General-iterative}). While there is no communication cost for a centralized implementation, the computation complexity is $O(nN^{2})$, and the storage cost is $O(N^2)$. Using $P$ processing nodes, the BaselineParallel algorithm - Algorithm \ref{euclid} - reduces computational complexity as well as storage cost by a factor of $P$.
In the BaselineParallel algorithm, in the off-line preprocessing step, the master node splits matrix $\mathbf{A}$ and $\mathbf{Q}$ equally horizontally as $\mathbf{A} = \begin{bmatrix}
\mathbf{A}_1^T & \ldots & \mathbf{A}^T_P
\end{bmatrix}^T, \mathbf{Q}=\begin{bmatrix}
\mathbf{Q}_1^T & \ldots & \mathbf{Q}_P^T
\end{bmatrix}^T$ and sends $\mathbf{A}_i, \mathbf{Q}_i$ to the $i$-th worker. In the online phase, the master node sends $\mathbf{x}^{(0)}$ and $\mathbf{y}$ to each worker to perform (\ref{eq:General-iterative}). The online phase of the algorithm is performed in $n$ rounds, each round corresponding to one iteration. In the $\ell$-th iteration, the master node sends $\mathbf{x}^{(\ell-1)}$ to all the worker nodes. Worker $i$ computes $\mathbf{A}_i \mathbf{x}^{(\ell-1)} + \mathbf{Q}_i \mathbf{y}$ and then sends it to the master node which aggregates the results of all workers to obtain $\mathbf{x}^{(\ell)}$.
{\small{\begin{algorithm}[h]
\caption{BaselineParallel($\mathbf{A},\mathbf{Q},n$) (Figure \ref{fig:0})
}\label{euclid}
\begin{algorithmic}[1]
\State \textbf{[${A}$] Offline computations}-\textbf{Master node:} Split matrices $\mathbf{A}$ and $\mathbf{Q}$ into $P$ equal-dimension submatrices such that $\mathbf{A}=\begin{bmatrix}
\mathbf{A}_1^T&
\ldots&
\mathbf{A}^T_P
\end{bmatrix}^T, \mathbf{Q}=\begin{bmatrix}
\mathbf{Q}_1^T&
\ldots&
\mathbf{Q}_P^T
\end{bmatrix}^T$. Then, send submatrices $\mathbf{A}_i$ and $\mathbf{Q}_i$,  $\mathbf{x}^{(0)}$ to the $i$th worker for $1\leq i\leq P$.
\State \textbf{[B] Online computations:} Send $\mathbf{y}$ to the $i$th worker.
 \State \textbf{For} ${k=1:n}$ repeat
 \State \quad \textbf{For} worker node $i=1$ to $P$ do
 \State \quad Compute $\mathbf{x}^{(k)}_{i}=\mathbf{A}_i\mathbf{x}^{(k-1)}+\mathbf{Q}_i\mathbf{y}$ and send $\mathbf{x}^{(k)}_{i}$ to the master node.
 \State \quad Master node aggregates $\mathbf{x}^{(k)}_{i}$'s to form $\mathbf{x}^{(k)}=\begin{bmatrix}
 {\mathbf{x}_{1}^{(k)}}&
 \ldots&
 {\mathbf{x}_{P}^{(k)}}
 \end{bmatrix}^T$ and then send $\mathbf{x}^{(k)}$ to each worker.
\end{algorithmic}
\end{algorithm}}}
For the BaselineParallel algorithm, the worker computational complexity is $O(nN^2/P)$ and the storage cost is $O(N^{2}/P)$. Note that in each round, the algorithm communicates one $N\times 1$ vector from the master node to each worker node, and one $\frac{N}{P}\times 1$ vector from each worker node to the master node. Since there are $n$ rounds, the communication complexity is $n \beta_1 + nN(1+P)/P \beta_2$.
\section{A new distributed linear inverse solver} \label{sec:parallel-extreme1}
\begin{figure}
  \begin{center}
    \includegraphics[scale=0.55]{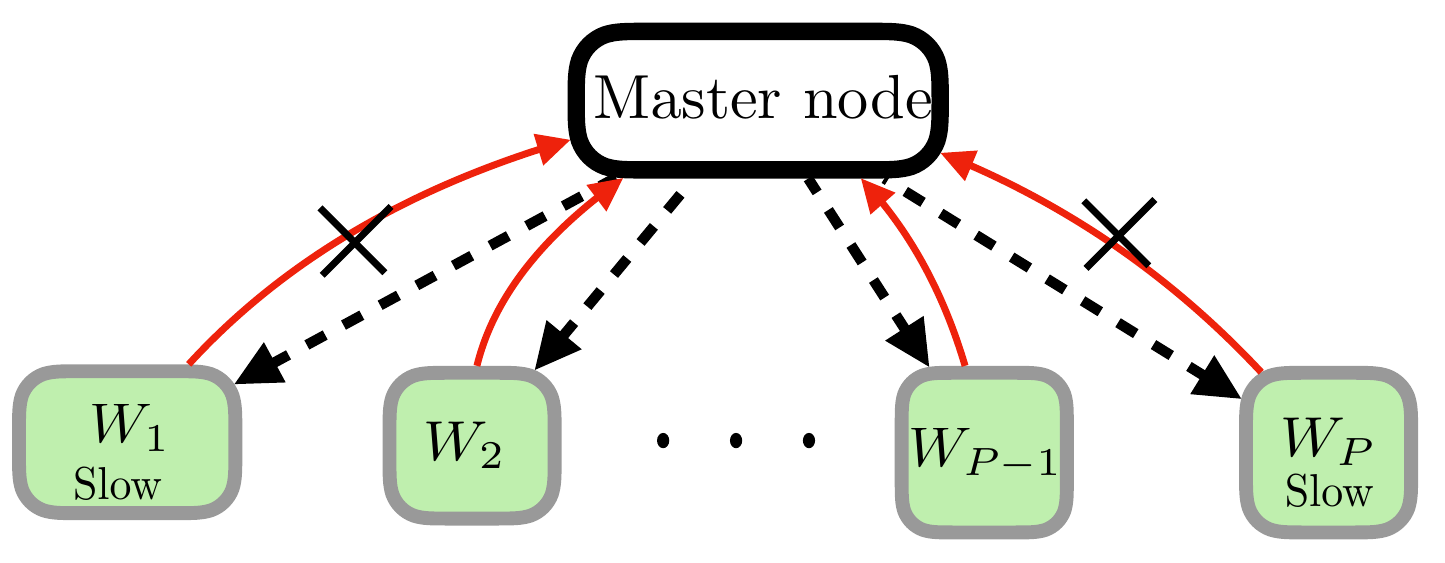}
  \end{center}
\caption{Straggler-resilient algorithm set up where master node waits only for the fastest $K$ worker node at each communication round.} \label{fig:3}
\end{figure}

        In Section \ref{sec:Propertiesofmatrixpoly}, we provide some preliminary results related to matrix polynomials that will be useful in our algorithm development. The results of Section \ref{sec:Propertiesofmatrixpoly} are a review of some results of \cite{MatDot}, as well as some new observations. Then in Section \ref{sec:polybasedalg1}, we describe the {PolyLin} algorithm and its costs. This algorithm needs only one round of communication and is resilient to certain number of stragglers as shown in Fig. \ref{fig:4}.

 \subsection{Properties of certain matrix polynomials}\label{sec:Propertiesofmatrixpoly}
 In this section, we begin by reviewing some relevant ideas and results of \cite{yu2017polynomial,MatDot,MatDot_Allerton}, which studied matrix multiplications. To begin with, consider matrix $\mathbf{B}$ which is split into submatrices such that $\mathbf{B}=\begin{bmatrix}
  \mathbf{\tilde{B}}_{0} & \mathbf{\tilde{B}}_{1}\end{bmatrix}=\begin{bmatrix}
  \mathbf{B}_{20}^T&
  \mathbf{B}_{21}^T
  \end{bmatrix}^T$ where for $j\in\{0,1\}$,  $\mathbf{\tilde{B}}_{j}$ and $\mathbf{B}_{2j}$ have dimension $N \times N / 2$ and $\frac{N}{2}\times N$ respectively. Now, let $p_\mathbf{\mathbf{B}_1}(\xi)= \mathbf{\tilde{B}}_{0} + \mathbf{\tilde{B}}_{1}\xi$ and $p_{\mathbf{B}_2}(\xi)= \mathbf{B}_{20} \xi + \mathbf{B}_{21}$. Note that $\mathbf{B}^2=\mathbf{\tilde{B}}_{0}\mathbf{B}_{20}+\mathbf{\tilde{B}}_{1}\mathbf{B}_{21} $ is the coefficient of $\xi$ of the matrix polynomial  $p_{\mathbf{B}_1}(\xi) p_{\mathbf{B}_2}(\xi) =(\mathbf{\tilde{B}}_{0} + \mathbf{\tilde{B}}_{1}\xi)(\mathbf{B}_{20} \xi + \mathbf{B}_{21})= \mathbf{\tilde{B}}_{1}\mathbf{B}_{20} \xi^2 + (\mathbf{\tilde{B}}_{0}\mathbf{B}_{20}+\mathbf{\tilde{B}}_{1}\mathbf{B}_{21} ) \xi + \mathbf{\tilde{B}}_{0}\mathbf{B}_{21}$. To implement the multiplication $\mathbf{B}^2\mathbf{z}$ in a distributed manner, where $\mathbf{z}$ is a $N\times 1$ vector on $P=3$ worker nodes, select three distinct real numbers $\xi_1, \xi_2, \xi_3$, and allow node $i$ to perform the multiplication $p_{\mathbf{B}_1}(\xi_i) p_{\mathbf{B}_2}(\xi_i) \mathbf{z}$. Then, the master node can interpolates the vector polynomial $p_{\mathbf{B}_1}(\xi_i) p_{\mathbf{B}_2}(\xi_i) \mathbf{z}$, and then finds $\mathbf{B}^2\mathbf{z}$ as the coefficient of $\xi$ of the interpolated polynomial. In general, if the matrix $\mathbf{B}$ is split into $m> 2$ parts (similarly both row and column wise), then by forming polynomials $p_{\mathbf{B}_1}(\xi), p_{\mathbf{B}_2}(\xi)$ of degree $m-1$ in a similar manner, the computation $\mathbf{B}^2\mathbf{z}$ can be performed over the results of any $2m-1$ worker nodes with the master node interpolating the degree $2m-2$ polynomial $p_{\mathbf{B}_1}(\xi_i) p_{\mathbf{B}_2}(\xi_i) \mathbf{z}$. Note interestingly that, if we set $\mathbf{B}=\mathbf{A}$ and $\mathbf{z} = \mathbf{x}^{(0)}$ then the above approach computes $\mathbf{A}^{2}\mathbf{x}^{(0)}$, which is one component of (\ref{eq:equivaldesc}) with $n=2$.

A generalization of this idea to multiply more than two matrices, as described in \cite{MatDot} is as follows. Suppose we want to compute $\mathbf{B}^4\mathbf{z}$. Now, note that $\mathbf{B}^4\mathbf{z}$ is the coefficient of $\xi^3$ in $\big(p_{\mathbf{B}_1}(\xi) p_{\mathbf{B}_2}(\xi)\big)\big(p_{\mathbf{B}_1}(\xi^2) p_{\mathbf{B}_2}(\xi^2)\big)\mathbf{z}$. To implement the multiplication $\mathbf{B}^4\mathbf{z}$ in a distributed manner using $P=7$ worker nodes, choose 7 distinct real numbers $\xi_1,\dots,\xi_7$ and let worker node $i$ perform $\big(p_{\mathbf{B}_1}(\xi_i) p_{\mathbf{B}_2}(\xi_i)\big)\big(p_{\mathbf{B}_1}(\xi_i^2) p_{\mathbf{B}_2}(\xi^2_i)\big)\mathbf{z}$. Finally, the master node similarly can recover $\mathbf{B}^4\mathbf{z}$ via interpolation, which is another component in (\ref{eq:equivaldesc}) by setting $\mathbf{B}=\mathbf{A}$ and $\mathbf{z}=\mathbf{x}^{(0)}$.

 We apply the above observations to our context. Assume $n$ is even, we split $\mathbf{A}$ both vertically and horizontally. Further, split $\mathbf{I}$ and  $\mathbf{Q}$ only horizontally as follows:
            \begin{align}
               \mathbf{A}&=\begin{bmatrix}
                                    \mathbf{A}_{10} & \mathbf{A}_{11} & \dots & \mathbf{A}_{1m-1}
                                    \end{bmatrix}=\begin{bmatrix}
                                             \mathbf{A}_{20} ^T&
                                             \mathbf{A}_{21}^T&
                                             \ldots&
                                             \mathbf{A}_{2\: m-1}^T 
                                             \end{bmatrix}^T\label{eq:splitAeven}\\
                          \mathbf{Q}&=\begin{bmatrix}
                                               \mathbf{Q}_{20}^T & \mathbf{Q}^T_{21} & \dots & \mathbf{Q}^T_{2\:m-1}
                                               \end{bmatrix}^T,
                                               \mathbf{I}=[\mathbf{I}_0^T,\dots,\mathbf{I}_{m-1}^T]^T   \label{eq:splitI}\end{align}
  where $\mathbf{I}$ is an identity matrix of dimension $N\times N$. Next, form the polynomials
 {{\small{ \begin{align}
  p_{\mathbf{A}_1}(\xi)=\sum_{j=0}^{m-1}\mathbf{A}_{1j}\xi^j,\:
  p_{\mathbf{A}_2}(\xi)=\sum_{j=0}^{m-1}\mathbf{A}_{2j}\xi^{m-1-j},
  p_{\mathbf{Q}_2}(\xi)=\sum_{j=0}^{m-1}\mathbf{Q}_{2j}\xi^{m-1-j}
  ,\:p_{\mathbf{I}}(\xi)=\sum_{j=0}^{m-1}\mathbf{I}_j\xi^j\label{eq:polI}
  \end{align}}}
  and set $p_{\mathbf{C}}(\xi)\triangleq p_{\mathbf{A}_1}(\xi)p_{\mathbf{A}_2}(\xi)$,  and $p_{\mathbf{D}}(\xi)\triangleq p_{\mathbf{A}_1}(\xi)p_{\mathbf{Q}_2}(\xi)$. While the dimension of $p_{\mathbf{A}_1}(\xi)$ is $N\times\frac{N}{m}$, the dimensions of $p_{\mathbf{A}_2}(\xi),p_{\mathbf{Q}_2}(\xi)$ and $p_{\mathbf{I}}(\xi)$ are $\frac{N}{m}\times N$.
Note that $\mathbf{A}^2$ are the coefficient of $\xi^{m-1}$ in $p_{\mathbf{C}}(\xi)$. Recall that the ordering of multiplication is important, so e.g. $\Pi_{i= i_0}^{i_1} \mathbf{M}_i=\mathbf{M}_{i_0} \mathbf{M}_{i_0-1} \dots \mathbf{M}_{i_1+1}\mathbf{M}_{i_1}$ for $i_0 > i_1$. For $l\geq 4$, let $P^{(l,m)}(\xi) \triangleq\left\{ \begin{array}{ll}
         p_{\mathbf{I}}(\xi^{m^{\lceil\frac{l}{2}\rceil-1}})\Pi_{i=\frac{l-1}{2}}^{2}p_{\mathbf{C}}(\xi^{m^{i-1}})p_{\mathbf{D}}(\xi)& \text{odd}\:l,\\
         \Pi_{i=\frac{l}{2}}^{2}p_{\mathbf{C}}(\xi^{m^{i-1}})p_{\mathbf{D}}(\xi) & \text{else}.
         \end{array}\right.$. Then, we have:
 \begin{lemma}\label{claim:evensumdot}
  For even $n$, $\mathbf{A}^{n}\mathbf{x}^{(0)}+\sum_{i=1}^{n}\mathbf{A}^{i-1}\mathbf{Q}\mathbf{y}$ is the coefficient of $\xi^{m^{\frac{n}{2}}-1}$ in the degree- $2m^{\frac{n}{2}}-2$ polynomial
 \begin{align}
 \eta(\xi,n)&=\Pi_{i=\frac{n}{2}}^{1} p_{\mathbf{C}}(\xi^{m^{i-1}})\mathbf{x}^{(0)}+\sum_{i=5}^{n}\xi^{m^{\frac{n}{2}}-m^{\lceil\frac{i}{2}\rceil}} P^{(i,m)}(\xi)\mathbf{y}\nonumber\\
 &+ \xi^{m^{\frac{n}{2}}-m^2}\big(p_{\mathbf{C}}(\xi^m)+p_{\mathbf{I}}(\xi^m) p_{\mathbf{A}_2}(\xi^m)\big)p_{\mathbf{D}}(\xi)\mathbf{y}+\xi^{m^{\frac{n}{2}}-m}(p_{\mathbf{D}}(\xi)+p_{\mathbf{I}}(\xi)p_{\mathbf{Q}_2}(\xi))\mathbf{y}\label{eq:final-interpolation}.
 \end{align}

 \end{lemma}
\begin{corollary}\label{corol:decodingcomplexity}
Let $K=2m^{\frac{n}{2}}-1$ and $\xi_1, \xi_2, \dots, \xi_{K}$ be distinct numbers. Then, there is an algorithm with complexity $O(NK\log^2K\log(\log K))$ that takes as input $\eta(\xi_1,n),\dots,\eta(\xi_{K},n)$ (polynomial in (\ref{eq:final-interpolation})) and outputs $\mathbf{A}^{n}\mathbf{x}^{(0)} + \sum_{i=1}^n \mathbf{A}^{i-1}\mathbf{Q}\mathbf{y}$.
\end{corollary}

The proof of Lemma \ref{claim:evensumdot} (which is in Appendix \ref{c}) combined with polynomial interpolation. Since interpolating a degree $d-1$ polynomial has complexity $O(d\log^2d\log(\log d))$ \cite{kedlaya2011fast} and noting that $\eta(\xi_i,n)$ are vectors of dimension $N\times 1$, the complexity mentioned in the corollary statement follows.\vspace{-0.5 em}
 \subsection{{PolyLin}:  A polynomial-evaluation based {fault-tolerant} distributed linear inverse solver}\label{sec:polybasedalg1}
 We present PolyLin  in Algorithm \ref{pts1} for the case when the number of iterations $n$ is even, The case of odd n is a bit more technically involved, since it involves extending Lemma \ref{claim:evensumdot} for this case, and is omitted in this submission.
 In Algorithm \ref{pts1}, worker node $i$ computes $\eta(\xi_i,n)$ in (\ref{eq:final-interpolation}) iteratively. We ensure that there are $2m^{n/2}-1$ evaluations of $\eta(\xi,n)$, so that the master node can obtain the worker node output and reconstruct (\ref{eq:equivaldesc}) based on Corollary \ref{corol:decodingcomplexity}. {We do this by setting $2m^{n/2}-1 = K$, that is $m={(\frac{K+1}{2})}^{\frac{2}{n}}$ where $K\leq P$.} 

 First note that $\eta(\xi,n)=\mathbf{R}^{(n)}+\mathbf{s}^{(n)}$ where $\mathbf{R}^{(n)}\triangleq\Pi_{i=\frac{n}{2}}^{1} p_{\mathbf{C}}(\xi^{m^{i-1}})\mathbf{x}^{(0)}$ and $
 \mathbf{s}^{(n)}\triangleq\sum_{i=5}^{n}\xi^{m^{\frac{n}{2}}-m^{\lceil\frac{i}{2}\rceil}} P^{(i,m)}(\xi)\mathbf{y}+ \xi^{m^{\frac{n}{2}}-m^2}\big(p_{\mathbf{C}}(\xi^m)+p_{\mathbf{I}}(\xi^m) p_{\mathbf{A}_2}(\xi^m)\big)p_{\mathbf{D}}(\xi)\mathbf{y}+\xi^{m^{\frac{n}{2}}-m}(p_{\mathbf{D}}(\xi)+p_{\mathbf{I}}(\xi)p_{\mathbf{Q}_2}(\xi))\mathbf{y}$. Therefore, to compute $\eta(\xi,n)$, each worker can compute separately $\mathbf{R}^{(n)}$ and $\mathbf{s}^{(n)}$
 iteratively and add them in final iteration.
 The variable $\mathbf{r}^{(i)}$ is used to compute $\mathbf{R}^{(n)}=\Pi_{i=\frac{n}{2}}^{1} p_{\mathbf{C}}(\xi^{m^{i-1}})\mathbf{x}^{(0)}$ iteratively. To this end, initially for $i=1$ each worker computes $\mathbf{r}^{(1)}=p_{\mathbf{A}_2}(\xi)\mathbf{x}^{(0)}$ in line 6 of the Algorithm \ref{pts1}. Then, given the evaluation polynomials, it computes $\mathbf{r}^{(i)}=\left\{ \begin{array}{ll}
          p_{\mathbf{A}_1}(\xi^{m^{\lceil\frac{i}{2}\rceil-1}})\mathbf{r}^{(i-1)}& \text{even}\:i,\\
          p_{\mathbf{A}_2}(\xi^{m^{\lceil\frac{i}{2}\rceil-1}})\mathbf{r}^{(i-1)} & \text{odd}\:i,
          \end{array}\right.$ which is shown in lines the 8 and 9 of algorithm. So, for even $n$, $\mathbf{R}^{(n)}=\mathbf{r}^{(n)}=\Pi_{i=\frac{n}{2}}^{1} p_{\mathbf{C}}(\xi^{m^{i-1}})\mathbf{x}^{(0)}$.
          To compute $\mathbf{s}^{(n)}$ iteratively, we need to compute $P^{(i,m)}$ at each iteration. Two variables $\mathbf{s}^{(i)},\mathbf{w}^{(i)}$ are used to compute $P^{(i,m)}$ at iteration $i$. The term $\xi^{m^{\frac{n}{2}}-m}(p_{\mathbf{D}}(\xi)+p_{\mathbf{I}}(\xi)p_{\mathbf{Q}_2}(\xi))\mathbf{y}$ is  computed in line 6 of the algorithm. Equations in line 8 and 9 imply that for $i\geq 2$,
$\mathbf{w}^{(i)}=\left\{ \begin{array}{ll}
         p_{\mathbf{A}_1}(\xi^{m^{\lceil\frac{i}{2}\rceil-1}})\mathbf{w}^{(i-1)}& \text{even}\:i,\\
         p_{\mathbf{A}_2}(\xi^{m^{\lceil\frac{i}{2}\rceil-1}})\mathbf{w}^{(i-1)} & \text{odd}\:i,
         \end{array}\right.$. The computation in line 9 then implies $\mathbf{s}^{(i)}=\mathbf{s}^{(i-2)}++\xi_l^{m^{\frac{n}{2}}-m^{\lceil\frac{i}{2}\rceil}}(\mathbf{w}^{(i)}+p_{\mathbf{I}}(\xi^{m^{\lceil\frac{i}{2}\rceil-1}}_l)\mathbf{w}^{(i-1)})$ as required. Therefore, after iteration $n$, $\mathbf{R}^{(n)}+\mathbf{s}^{(n)}=\eta(\xi,n)$.

\begin{theorem}\label{thm:dist1}
If $n\geq \log_{\frac{1}{|\sigma_1|}}\frac{N\underset{1\leq i\leq N}{\max}|\alpha_i|}{\epsilon}$, the error associated with PolyLin in computing the solution of inverse problem is at most $\epsilon$.
\end{theorem}

In Algorithm \ref{pts1} eventually the $l$-th processing node computes $\mathbf{R}^{(n)}+\mathbf{s}^{(n)}=\eta(\xi_l)$. So,
using Lemma \ref{claim:evensumdot} and Corollary \ref{corol:decodingcomplexity} and noting that master node receives $K=2m^{\frac{n}{2}}-1$ distinct polynomial evaluation vectors, it can recover $\mathbf{A}^n\mathbf{x}^{(0)}+\sum_{i=1}^{n}\mathbf{A}^{i-1}\mathbf{Q}\mathbf{y}$ and the proof of the Theorem \ref{thm:dist1} follows.

{\small{\begin{algorithm}[h]
     \caption{PolyLin $(\mathbf{A},\mathbf{Q},P,K,\mathbf{x}^{(0)},\mathbf{y},n)$ \label{pts1}}
     \begin{algorithmic}[1]
    \State \textbf{One-time preprocessing step: Input: $\mathbf{A}_{N\times N}$}, given the number of iterations $n$, vector $\mathbf{y}$, matrix $\mathbf{Q}$ (\textbf{even $\mathbf{n}$ and matrix $\mathbf{Q}$}) initial point $\mathbf{x}^{(0)}$ and $P-K$ stragglers.
       \State \textbf{Master node:} Split $\mathbf{A}$ and $\mathbf{Q}$ using (\ref{eq:splitAeven}) and (\ref{eq:splitI}). Set $\xi_1,\xi_2,\dots,\xi_P$ be arbitrary distinct numbers. The master node sends $\mathbf{\xi^{(0)}}$, $\mathbf{y}$, $p_{\mathbf{Q}_2}(\xi_l)$, $p_{\mathbf{I}}(\xi_l^{m^{i-1}})$ (odd $i$), $p_{\mathbf{A}_1}(\xi_l^{m^{i-1}})$ and $p_{\mathbf{A}_2}(\xi_l^{m^{i-1}})$ in (\ref{eq:polI}) for $1\leq i\leq { \frac{{n}} {2}}$ to worker $l\in\{1,2,\dots,P\}$ where $\mathbf{m}=(\frac{K+1}{2})^{\frac{2}{n}}$ and $K\leq P$.
            \State \textbf{Online computations:}
       \State \textbf{For}  ${l=1}$ to ${P}$ do
       \State  \quad \textbf{Input:}  $\mathbf{\xi^{(0)}}$, $\mathbf{y}$,  $p_{\mathbf{Q}_2}(\xi_l)$, $p_{\mathbf{I}}(\xi_l^{m^{i-1}})$, $p_{\mathbf{A}_2}(\xi_l^{m^{i-1}})$ for odd $i$ and $p_{\mathbf{A}_2}(\xi_l^{m^{i-1}})$ for even $i$ where $1\leq i\leq { \frac{{n}} {2}}$.
       \State \textbf{Compute}: $
         \mathbf{r}^{(1)}=p_{\mathbf{A}_2}(\xi_l)\mathbf{x}^{(0)}, \mathbf{r}^{(2)}=p_{\mathbf{A}_1}(\xi_l)\mathbf{r}^{(1)},
       \mathbf{w}^{(1)}=p_{\mathbf{Q}_{2}}(\xi_l)\mathbf{y},\mathbf{w}^{(2)}=p_{\mathbf{A}_1}(\xi_l)\mathbf{w}^{(1)},\mathbf{s}^{(2)}=\xi_l^{m^{\frac{n}{2}}-m}(\mathbf{w}^{(2)}+p_{\mathbf{I}}(\xi_l)\mathbf{w}^{(1)})$
       \State \textbf{For} $i=3:n$
        \State \quad \textbf{If} $i$ is odd, compute $
        \mathbf{r}^{(i)}=p_{\mathbf{A}_2}(\xi_l^{m^{\lceil\frac{i}{2}\rceil-1}})\mathbf{r}^{(i-1)},\mathbf{w}^{(i)}=p_{\mathbf{A}_2}(\xi^{m^{\lceil\frac{i}{2}\rceil-1}}_l)\mathbf{w}^{(i-1)},$
        \State \quad \textbf{Else} compute,
  $\mathbf{r}^{(i)}=p_{\mathbf{A}_1}(\xi_l^{m^{\lceil\frac{i}{2}\rceil-1}})\mathbf{r}^{(i-1)},\mathbf{w}^{(i)}=p_{\mathbf{A}_1}(\xi^{m^{\lceil\frac{i}{2}\rceil-1}}_l)\mathbf{w}^{(i-1)},\mathbf{s}^{(i)}=\mathbf{s}^{(i-2)}+\xi_l^{m^{\frac{n}{2}}-m^{\lceil\frac{i}{2}\rceil}}(\mathbf{w}^{(i)}+p_{\mathbf{I}}(\xi^{m^{\lceil\frac{i}{2}\rceil-1}}_l)\mathbf{w}^{(i-1)})$
            \State \quad \textbf{Output:} $\mathbf{r}^{(n)}(\xi_l)+\mathbf{s}^{(n)}(\xi_l)$
       \State \textbf{Post-Processing in the master node:}
       Interpolate the output from the fastest $K$ workers' results.
     \end{algorithmic}
     \end{algorithm}}}\vspace{-0.5 em}

\section{Trade-off Between Communication and Computation Costs}\label{sec:tradeoff}
While communication cost of the PolyLin, $\beta_1+N\beta_2$ is smaller than the communication cost of BaselineParallel, $n\beta_1+nN(\frac{1+P}{P})\beta_2$, the computational cost of the PolyLin, $O({nN^2}/{{(\frac{K+1}{2})}^{\frac{2}{n}}})$ {where $K\leq P$} is larger than the computational cost of the BaselineParallel, $O(\frac{nN^2}{P})$.
Next, we present the MRPolylin algorithm that is a generalization that trades-off between these two extremes.

\textbf{Algorithm and its description:} MRPolyLin is a parametrized version of PolyLin over some integer $\ell$, which is smaller than the number of iterations $n$ and $\ell|n$. MRPolyLin divides $n$ iterations into $\ell$ phases and in each phase we conduct PolyLin for $n/\ell$ iterations. 
As compared with the PolyLin algorithm, the number of communication rounds changes from $1$ to $\ell$, and the worker computational cost reduces to $O({nN^2}/{{(\frac{K+1}{2})}^{\frac{2\ell}{n}}})$ {where $K\leq P$}.
\begin{algorithm}[h]
     \caption{MRPolyLin $(\mathbf{A},\mathbf{Q},P,K,\mathbf{x}^{(0)},\mathbf{y},n,\ell)$ \label{alg:pts3r}}
     \begin{algorithmic}[1]
    \State \textbf{Input:} $\mathbf{A}_{N\times N}$, the required number of iterations $n$ and $P$ workers, and vector $\mathbf{y}$, initial point $\mathbf{x}^{(0)}$ and $\ell$ rounds of communication.
       \State \textbf{For} ${j=1}$ to ${\ell}$ repeat: $\mathbf{x}^{(j)}=$
       PolyLin$(\mathbf{A},\mathbf{Q},P,K,\mathbf{x}^{(j-1)},\mathbf{y},\frac{n}{\ell})$.
      \end{algorithmic}
     \end{algorithm}

\begin{table*}
\centering
\caption{In this table $n=\log_{\frac{1}{|\sigma_1|}}\frac{N\underset{1\leq i\leq N}{\max}|\alpha_i|}{\epsilon}, n^*=n\log{1/|\sigma_1|}-\log(N\max|\alpha_i|) $.} \label{table:comp}
\resizebox{1.1\linewidth}{!}{
\begin{tabular}{c|ccccc}
  \hline
  Strategy & Computational cost & Storage cost & {Straggler resilience}& Communication cost\\
  \hline
  \\
{DSVRG} & {$O(\frac{n^*N^2}{P}+n^*N\kappa)$} & $O(\frac{N^2}{P})$ & {0}& {$\beta_1n^*(1+\frac{P\kappa}{N})+\beta_2O(n^*N(1+\frac{\kappa}{N}))$}
\\
            \hline
            \\
{DASVRG} & {$O(\frac{n^*N^2}{P}+n^*\sqrt{\frac{N^3}{P}\kappa})$} &  $O(\frac{N^2}{P})$ & {0}& {$\beta_1n^*(1+\sqrt{\frac{P\kappa}{N}})\log(1+\frac{P\kappa}{N})+\beta_2O(n^*N(1+\sqrt{\frac{P\kappa}{N}}))$}
            \\
                        \hline
                        \\
{Karakus \cite{karakus2017straggler}} & {$O(\frac{nN^2}{K})$} &  {$O(\frac{N^2}{K})$} & {$P-K$}& {$\beta_1n+\beta_2n(2N)$}
            \\
                        \hline
                        \\
 BaselineParallel & $O(\frac{nN^2}{P})$ & $O(\frac{N^2}{P})$ & {0}& $\beta_1n+\beta_2n(\frac{1+P}{P})N$ \\
 \hline
 \\
  PolyLin & $O({nN^2}/{{(\frac{K+1}{2})}^{\frac{2}{n}}}),K\leq P$ &  $O((n+1){N^2}/{(\frac{K+1}{2})}^{\frac{2}{n}})$& {$P-K$} & $\beta_1+\beta_2(2N)$
  \\
  \hline
  \\
  MRPolyLin & $O({nN^2}/{{(\frac{K+1}{2})}^{\frac{2\ell}{n}}}),K\leq P$ & $O((\frac{n+\ell}{\ell}){N^2}/{(\frac{K+1}{2})}^{\frac{2\ell}{n}})$& {$P-K$} & $\beta_1\ell+\beta_2 (2\ell N)$
   \\
    \hline
\end{tabular}}
\end{table*}\vspace{-0.5 em}


\paragraph{Comparison of various algorithms in Table \ref{table:comp}:} While the pre and post-processing cost of DSVRG, DASVRG and BaselineParallel are $O(N)$ and $O(nN)$ respectively, pre and post-processing cost of  MRPolyLin are $O(nPN^2/\ell)$ and $O(\ell N K\log^2K\log(\log K))$ respectively (Letting $\ell=1$ gives the pre/post processing cost of PolyLin). {As can be seen in Table \ref{table:comp}, in DSVRG and DASVRG algorithms (and others in the Appendix \ref{e}), the number of communication rounds is proportional to $n^*=n\log{1/|\sigma_1|}-\log(N\max\alpha_i)=\log(\frac{1}{\epsilon})$. However, the communication cost of BaselineParallel is proportional to $n=\log_{\frac{1}{|\sigma_1|}}\frac{N\underset{1\leq i\leq N}{\max}|\alpha_i|}{\epsilon}$. Note that {when the error-requirement $\epsilon$ is small enough so that $n\geq \frac{\log(N\max\alpha_i)}{\log(\frac{1}{\sigma_1})-1}$, our baseline algorithm is comparable to DSVRG and DASVRG schemes in terms of computational and communication costs.} Moreover, in one extreme while PolyLin requires one round of communication for solving the linear inverse problem at the cost of higher computational cost, MRPolyLin introduces a trade-off between communication and computation cost that can not be achieved by competing schemes in the Table \ref{table:comp}. {Finally, we report that PolyLin, MRPolyLin and the suggested algorithm in (\cite{karakus2017straggler}), tolerate $P-K$ stragglers.}
\section{{Experiments}}\label{experiment}
\begin{figure}
  \begin{center}
    \includegraphics[scale=0.55]{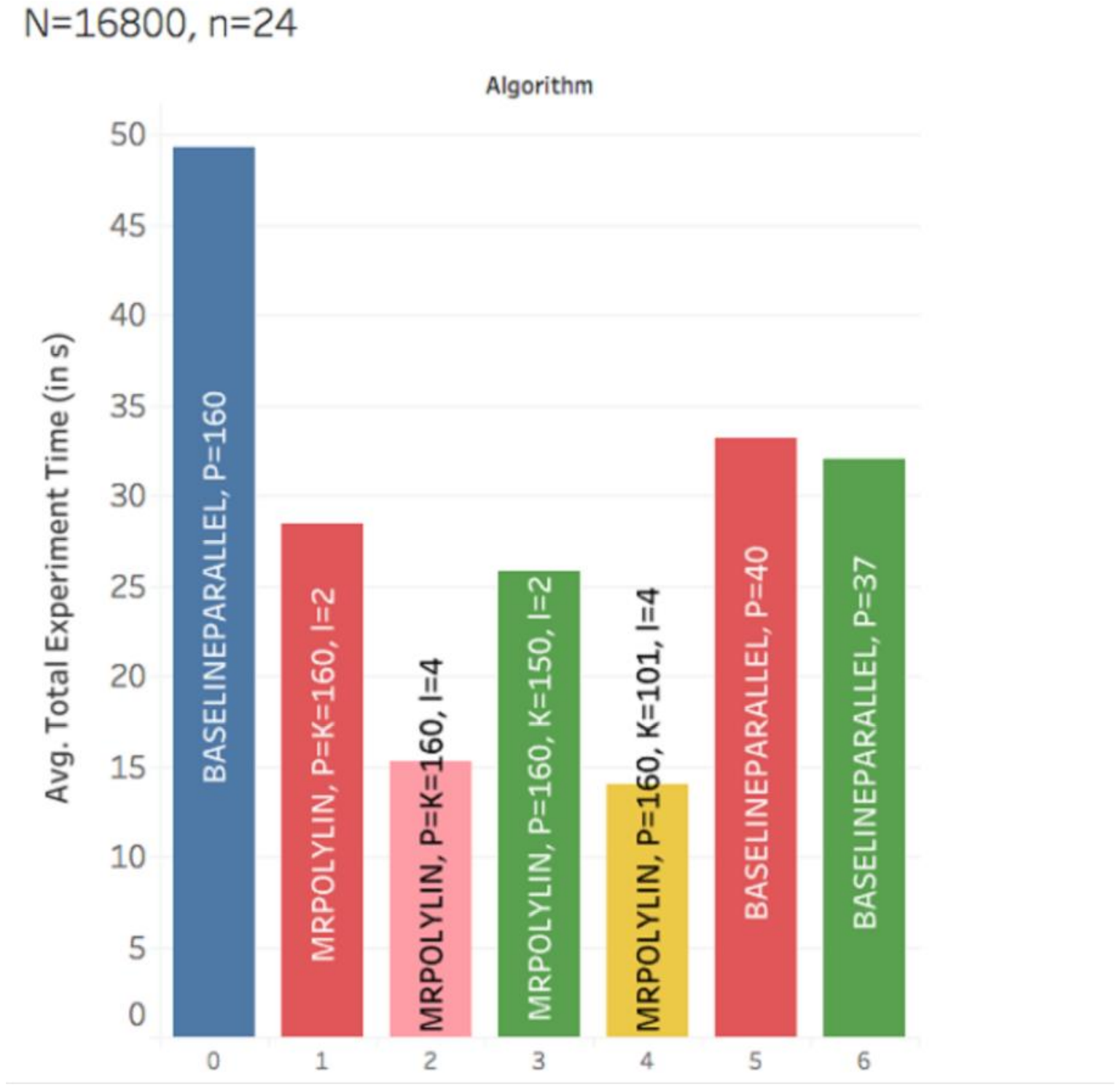}
  \end{center}
\vspace{-0.5 em}
\caption{Comparison of various algorithms.} \label{fig:4}\vspace{-0.5 em}
\end{figure}

We implemented our algorithms in a distributed computation prototyping framework built in Java and deployed on AWS EC2 cloud. We used a cluster of t2.medium instances and a random data matrix of dimension $\mathbf{A}_{168000\times 168000}$ for all the experiments. Note all the algorithms achieve same error associated with $n=24$ of communication of BaselineParallel algorithm using $P=160$. The total expectation time includes, communication and computational cost in addition to decoding cost. In Fig. \ref{fig:4} while algorithms 0, 5 and 6 indicate BaselineParallel using $P=160,40,37$ respectively, algorithms 1, 3 and 2, 4 indicate MRPolyLin with $\ell=12$ and $\ell=6$ rounds of communication with $P=160$. While algorithms 1 and 2 wait for all of the machines to complete their tasks (i.e., $K=P$), algorithms 3 and 4 wait for the fastest 110 and 101 machines to complete their task.  Algorithms with the same color have the same computational and storage costs. Comparing algorithms 0 to 4, observe that fewer communication rounds leads to faster completion time, despite computational load in our experiments. Compared to the baseline algorithms (5 and 6) with the same storage and computational complexity, our algorithms (1 and 3) achieve a speed up of around 20\%. Finally, comparing algorithms 1 $(K=P)$ and 3 ($P-K=10$) and algorithms 2 and 4 ($P-K=59$) we observe the speed up of 8\% at most over parallel schemes (where $K=P$).

\bibliographystyle{abbrv}
\bibliography{example_paper}

\begin{appendices}

\section{Brief explanation for Jacobi and Gradient Decent}\label{a}
\emph{ Jacobian Method for Square System:} For a square
matrix $\mathbf{M}$, we can write $\mathbf{M} = \mathbf{D}+\mathbf{L}$, where $\mathbf{D}$ is a diagonal matrix and $\mathbf{L}$ is a matrix whose diagonal entries are $0$. The Jacobian iteration is
 $\mathbf{x}^{(n+1)}=\mathbf{D}^{-1}(\mathbf{y}-\mathbf{L}\mathbf{x}^{(n)})$.
 For $\mathbf{D}$ and $\mathbf{L}$, the
 computation result converges to the true solution, $\mathbf{x}=\mathbf{M}^{-1}\mathbf{y}$.

 \emph{Gradient Descent Method:} For non-square matrices, the gradient descent solution has the form $\mathbf{x}^{(n+1)}=((1-\lambda)\mathbf{I}-\delta\mathbf{M}^T\mathbf{M})\mathbf{x}^{(n)}+\delta\mathbf{M}^T\mathbf{y}$
 where $\delta$ is an appropriate step-size.
\section{Error bound lemma proof}\label{b}
We derive a bound on the error of the iterative method as a function of number of iterations $n$. For simplicity, we assume $\mathbf{A}$ is diagonalizable and full rank\footnote{The case when $\mathbf{A}$ is not diagonalizable can also be analyzed with the Jordan decomposition.}. That is,
$\mathbf{e}^{(0)}=\alpha_1\mathbf{u}_1 + \alpha_2\mathbf{u}_2 + \dots + \alpha_N \mathbf{u}_N$, where $\mathbf{u}_1,\dots, \mathbf{u}_N$ are eigenvectors of matrix $\mathbf{A}$. Then,
$\mathbf{e}^{(1)}=\sigma_1\alpha_1\mathbf{u}_1 + \sigma_2\alpha_2\mathbf{u}_2 + \dots + \sigma_m\alpha_N \mathbf{u}_N$ where $\sigma_i$ are eigenvalues of $\mathbf{A}$ such that $|\sigma_N|\leq\dots\leq |\sigma_2|\leq |\sigma_1|<1$. Similarly, we have
$
\mathbf{e}^{(n)}=\mathbf{A}^n x^{(0)}=\sigma_1^n\alpha_1 \mathbf{u}_1 + \sigma_2^n\alpha_2\mathbf{u}_2 + \dots + \sigma_m^n\alpha_N \mathbf{u}_N
$.
   \begin{claim}[\textbf{Bound on error}]\label{}
   Letting $\epsilon(n)=\|\mathbf{e}^{(n)}\|$ and given that {$|\sigma_1|<1$}, the upper bound on the required number of iterations, denoted by ${n^*}$, is
   $$\log_{\frac{1}{\sigma_1}}\frac{rank(\mathbf{A})\underset{1\leq i\leq rank(\mathbf{A})}{\max}\alpha_i}{\epsilon}.$$
   \end{claim}
   \begin{proof}
   \begin{align}
   \epsilon(n)&=\|\sigma_1^{{n}}\alpha_1u_1 + \sigma_2^{{n}}\alpha_2u_2 + \dots + \sigma_m^{n}\alpha_m u_m\|\nonumber\\
   &\leq \|\sigma_1^{n}\alpha_1u_1 \|+\| \sigma_2^{n}\alpha_2u_2\| + \dots +\| \sigma_m^{n}\alpha_m u_m\|\nonumber\\
   &=|\sigma_1|^{n}|\alpha_1|+|\sigma_2|^{{n}}\alpha_2| + \dots + |\sigma_m|^{n}|\alpha_m|\nonumber\\
   &\stackrel{(a)}{\leq}m|\sigma_1|^{n}\underset{1\leq i\leq m}{\max}{|\alpha_i}|=\epsilon 
   \end{align}
   where in (a) we use the assumption $|\sigma_m|\leq\dots\leq|\sigma_2|<|\sigma_1|<1$ (due to assumption that eigenvalues of matrix $\mathbf{A}$ are strictly smaller than 1) in which $m=rank(\mathbf{A})$ then, if $n\geq{n^*}=\log_{\frac{1}{|\sigma_1|}}\frac{m\underset{1\leq i\leq m}{\max}\alpha_i}{\epsilon}$, we have $\epsilon(n)\leq \epsilon$.
   \end{proof}
  \section{Proof of Lemma \ref{claim:evensumdot}}\label{c}

We recall some relevant definitions of Section \ref{sec:Propertiesofmatrixpoly}. Split $\mathbf{A}_{N\times N}$ both vertically and horizontally and $\mathbf{I}_{N\times N}$ and  $\mathbf{Q}_{N\times N}$ only horizontally as following:
                    \begin{align}
                       \mathbf{A}&=\begin{bmatrix}
                                            \mathbf{A}_{10} & \mathbf{A}_{11} & \dots & \mathbf{A}_{1m-1}
                                            \end{bmatrix}=\begin{bmatrix}
                                                     \mathbf{A}_{20} ^T&
                                                     \mathbf{A}_{21}^T&
                                                     \ldots&
                                                     \mathbf{A}_{2\: m-1}^T 
                                                     \end{bmatrix}^T 
                                                     \\
                                  \mathbf{Q}&=\begin{bmatrix}
                                                       \mathbf{Q}_{20}^T & \mathbf{Q}^T_{21} & \dots & \mathbf{Q}^T_{2\:m-1}
                                                       \end{bmatrix}^T ,
                                                       \mathbf{I}_{N\times N}=[\mathbf{I}_0^T,\dots,\mathbf{I}_{m-1}^T]^T   
                                                       \end{align}
          Also recall the following matrix polynomials
          \begin{align}
          p_{\mathbf{A}_1}(x)=\sum_{j=0}^{m-1}\mathbf{A}_{1j}x^j,\:
          p_{\mathbf{A}_2}(x)=\sum_{j=0}^{m-1}\mathbf{A}_{2j}x^{m-1-j},
          p_{\mathbf{Q}_2}(x)=\sum_{j=0}^{m-1}\mathbf{Q}_{2j}x^{m-1-j}
          ,\:p_{\mathbf{I}}(x)=\sum_{j=0}^{m-1}\mathbf{I}_jx^j. 
          \end{align}
          Using the above polynomials, we define polynomials
          \begin{align}
          p_{\mathbf{C}}(x)\triangleq p_{\mathbf{A}_1}(x) p_{\mathbf{A}_2}(x),\quad p_{\mathbf{D}}(x)\triangleq p_{\mathbf{A}_1}(x)p_{\mathbf{Q}_2}(x).
          \end{align}
         Note that the dimensions of $p_{\mathbf{A}_1}(x)$ is $N\times\frac{N}{m}$ and the dimensions of $p_{\mathbf{A}_2}(x),p_{\mathbf{Q}_2}(x)$ and $p_{\mathbf{I}}(x)$ are all $\frac{N}{m}\times N$. The polynomial $p_{\mathbf{C}}(x)$ is used to construct partial results of $\mathbf{A}^2$ and $p_{\mathbf{D}}(x)$ is used to construct that of $\mathbf{AQ}$. Now, we restate Lemma \ref{claim:multipleMatdot} from the main draft as follows. Let for $l\geq 5$:
                    \begin{align}P^{(l,m)}(\xi) \triangleq\left\{ \begin{array}{ll}
                              \Pi_{i=\frac{l}{2}}^{2}p_{\mathbf{C}}(\xi^{m^{i-1}})p_{\mathbf{D}}(\xi) & \text{even}\:l,\\
                              p_{\mathbf{I}}(\xi^{m^{\lfloor\frac{l}{2}\rfloor}})\Pi_{i=\frac{l-1}{2}}^{2}p_{\mathbf{C}}(\xi^{m^{i-1}})p_{\mathbf{D}}(\xi)& \text{odd}\:l.
                              \end{array}\right.\label{eq:Pnm[x]}\end{align}

         To use these definitions, we need a variant of a Lemma from \cite{MatDot} which is stated as follows:

     \begin{lemma}[\cite{MatDot}]\label{claim:multipleMatdot}
            For $l\geq 4$ we have:
            \begin{itemize}
            \item[(i)] If $l$ is even, $\mathbf{A}^{l-1}\mathbf{Q}$ is the coefficient of $\xi^{m^{\frac{l}{2}}-1}$ in $\Pi_{i=\frac{l}{2}}^{2}p_{\mathbf{C}}(\xi^{m^{i-1}})p_{\mathbf{D}}(\xi)$.
            \item[(ii)] If $l$ is odd, $\mathbf{A}^{l-1}\mathbf{Q}$ is the coefficient of  $\xi^{m^{\frac{l+1}{2}}-1}$ in $p_{\mathbf{I}}(\xi^{m^{\frac{l-1}{2}}})\Pi_{i=\frac{l-1}{2}}^{2}p_{\mathbf{C}}(\xi^{m^{i-1}})p_{\mathbf{D}}(\xi)$.
            \end{itemize}
            \end{lemma}

         Since $P^{(l,m)}(\xi)$ in (\ref{eq:Pnm[x]}) is a polynomial of $\xi$, it can be written as
             $P^{(l,m)}(\xi)=\sum_{i=0}^{2m^{\lceil\frac{l}{2}\rceil}-2}q_i^{(l)}\xi^i$. Then, the coefficient of $q^{(n)}_{m^{\lceil\frac{l}{2}\rceil}-1}$ is $\mathbf{A}^{l-1}\mathbf{Q}$ because of Lemma \ref{claim:multipleMatdot}. Also note that for even $l$, the degree of $P^{(l,m)}(\xi)$ is equal to the degree of $P^{(l-1,m)}(\xi)$ because we have put the extra term $p_{\mathbf{I}}(\xi^{m^{\lfloor\frac{l}{2}\rfloor}})$ in \eqref{eq:Pnm[x]} to make $P^{(l,m)}(\xi)$ and $P^{(l-1,m)}(\xi)$ match in degree. This property will be useful later in the proof. Based on Lemma \ref{claim:multipleMatdot} we have:
  \begin{corollary}\label{corol:pairwise}
  For even $l$, $\mathbf{A}^{l-1}\mathbf{Q}+\mathbf{A}^{l-2}\mathbf{Q}$ is the co-efficient of $\xi^{m^{\frac{l}{2}}-1}$ in $$\Pi_{i=\frac{l}{2}}^{2}p_{\mathbf{C}}(\xi^{m^{i-1}})p_{\mathbf{D}}(\xi) +p_{\mathbf{I}}(\xi^{m^{\lfloor\frac{l}{2}\rfloor}})\Pi_{i=\frac{l-1}{2}}^{2}p_{\mathbf{C}}(\xi^{m^{i-1}})p_{\mathbf{D}}(\xi).$$
  \end{corollary}

          \begin{lemma}\label{claim:summatdot}
            For even $n$ and even $t$ such that $4\leq t< n$ and any vector $\mathbf{y}$, $\mathbf{A}^{t-1}\mathbf{Q}\mathbf{y}+\mathbf{A}^{t-2}\mathbf{Q}\mathbf{y}$ is the coefficient of $\xi^{m^{\frac{n}{2}}-1}$ in $\xi^{m^{\frac{n}{2}}-m^{\frac{t}{2}}}\big(P^{(t,m)}(\xi)+P^{(t-1,m)}(\xi)\big)\mathbf{y}$, and also  $\deg(\xi^{m^{\frac{n}{2}}-m^{\frac{t}{2}}}\big(P^{(t,m)}(\xi)+P^{(t-1,m)}(\xi)\big))=m^{\frac{n}{2}}+m^{\frac{t}{2}}-2$.
            \end{lemma}
            \begin{proof}
             From Lemma \ref{claim:multipleMatdot} and Corollary \ref{corol:pairwise} we know that $\mathbf{A}^{t-1}\mathbf{Q}+\mathbf{A}^{t-2}\mathbf{Q}$ is a co-efficient of $P^{(t,m)}(\xi)+P^{(t-1,m)}(\xi)$ and since $n$ and $t$ are even, we have
             \begin{align*}
             & \xi^{m^{\frac{n}{2}}-m^{\frac{t}{2}}} \big(P^{(t,m)}(\xi)+P^{(t-1,m)}(\xi)\big)\mathbf{y}\\
             &\quad=\xi^{m^{\frac{n}{2}}-m^{\frac{t}{2}}}(\sum_{i=0}^{m^{\frac{t}{2}}-2}q_i^{(t)}\xi^i\mathbf{y}+\big(\mathbf{A}^{t-1}\mathbf{Q}+\mathbf{A}^{t-2}\mathbf{Q}\big)\mathbf{y}\xi^{m^{\frac{t}{2}}-1}+\sum_{i={m^{\frac{t}{2}}}}^{2m^{\frac{t}{2}}-2}(q_i^{(t)}\xi^i)\mathbf{y})\\
             &\quad=\sum_{i=m^{\frac{n}{2}}-m^{\frac{t}{2}}}^{m^{\frac{n}{2}}-1}(q_i^{(t)}\mathbf{y})\xi^i+\big(\mathbf{A}^{t-1}\mathbf{Q}\mathbf{y}+\mathbf{A}^{t-2}\mathbf{Q}\mathbf{y}\big)\xi^{m^{\frac{n}{2}}-1}+\sum_{i={m^{\frac{n}{2}}}}^{m^{\frac{n}{2}}+m^{\frac{t}{2}}-2}(q_i^{(t)}\mathbf{y})\xi^i
             \end{align*}
             because $m^{\frac{n}{2}}+m^{\frac{t}{2}}-2<2m^{\frac{n}{2}}-2$, $\deg(P^{(n,m)}(\xi))=\deg(P^{(n,m)}(\xi)+\xi^{m^{\frac{n}{2}}-m^{\frac{t}{2}}}\big(P^{(t,m)}(\xi)+P^{(t-1,m)}(\xi)\big))$ and therefore $\mathbf{A}^{t-2}\mathbf{Q}\mathbf{y}+\mathbf{A}^{t-1}\mathbf{Q}\mathbf{y}+\mathbf{A}^{n-1}\mathbf{Q}\mathbf{y}$ is the coefficient of $\xi^{m^{\frac{n}{2}}-1}$ in $P^{(n,m)}(\xi)\mathbf{y}_1+\xi^{m^{\frac{n}{2}}-m^{\lceil\frac{t}{2}\rceil}}\big(P^{(t,m)}(\xi)+P^{(t-1,m)}(\xi)\big)\mathbf{y}_2$.
             \end{proof}

   The proof of Lemma \ref{claim:evensumdot}  is based on the multiple applications of Lemma \ref{claim:multipleMatdot} combined with the  Lemma \ref{claim:summatdot}. For the sake of readability, we restate Lemma \ref{claim:evensumdot} from the main draft as follows:
   \begin{lemma} 
      For even $n$, $\mathbf{A}^{n}\mathbf{x}^{(0)}+\sum_{i=1}^{n}\big(\mathbf{A}^{i-1}\mathbf{Q}\big)\mathbf{y}$ is the coefficient of $\xi^{m^{\frac{n}{2}}-1}$ in the degree $2m^{\frac{n}{2}}-2$ polynomial
     \begin{align}
     &\eta(\xi,n)=\Pi_{i=\frac{n}{2}}^{1} p_{\mathbf{C}}(\xi^{m^{i-1}})\mathbf{x}^{(0)}+\sum_{i=5}^{n}\xi^{m^{\frac{n}{2}}-m^{\lceil\frac{i}{2}\rceil}} P^{(i,m)}(\xi)\mathbf{y}\nonumber+ \xi^{m^{\frac{n}{2}}-m^2}\big(p_{\mathbf{C}}(\xi^m)+p_{\mathbf{I}}(\xi^m) p_{\mathbf{A}_2}(\xi^m)\big)p_{\mathbf{D}}(\xi)\mathbf{y}\nonumber\\&\quad+\xi^{m^{\frac{n}{2}}-m}(p_{\mathbf{D}}(\xi)+p_{\mathbf{I}}(\xi)p_{\mathbf{Q}_2}(\xi))\mathbf{y}. 
     \end{align}
     \end{lemma}
 \begin{proof}
    Since \begin{align*}
    \mathbf{A}^{n}\mathbf{x}^{(0)}+\sum_{i=1}^{n}\big(\mathbf{A}^{i-1}\mathbf{Q}\big)\mathbf{y}=\mathbf{A}^{n}\mathbf{x}^{(0)}+\underbrace{(\mathbf{A}^{n-1}\mathbf{Q}y+\mathbf{I}\mathbf{A}^{n-2}\mathbf{Q}y)}_{\mathbf{G}_{\frac{n}{2}}}+\underbrace{(\mathbf{A}^{n-3}\mathbf{Q}y+\mathbf{I}\mathbf{A}^{n-4}\mathbf{Q}y)}_{\mathbf{G}_{\frac{n}{2}-1}}+\dots+\underbrace{(\mathbf{A}\mathbf{Q}y+\mathbf{I}\mathbf{Q}y)}_{\mathbf{G}_1}
    \end{align*}
   $\mathbf{G}_1,\dots,\mathbf{G}_{\frac{n}{2}-1},\mathbf{G}_{\frac{n}{2}}, \mathbf{A}^n\mathbf{x}^{(0)}$ are the co-efficient of $\xi^{m^{\frac{n}{2}}-1}$ in $\xi^{m^{\frac{n}{2}}-m}(p_{\mathbf{D}}(\xi)+p_{\mathbf{I}}(\xi)p_{\mathbf{Q}_2}(\xi))\mathbf{y},\dots,\xi^{m^{\frac{n}{2}}-m^{\lceil\frac{n-2}{2}\rceil}} \big(P^{(n-2,m)}(\xi)+P^{(n-3,m)}(\xi)\big)\mathbf{y}, \big(P^{(n,m)}(\xi)+P^{(n-1,m)}(\xi)\big)\mathbf{y}$ and $\Pi_{i=\frac{n}{2}}^{1} p_{\mathbf{C}}(\xi^{m^{i-1}})\mathbf{x}^{(0)}$ respectively by applying Lemma \ref{claim:summatdot} and Corollary \ref{corol:pairwise}. This completes the proof.
\end{proof}
\section{Complexity analysis of PolyLin algorithm}\label{d}
\emph{Communication complexity}: In PolyLin algorithm there is only one round of communication. In this round, the master node first sends $\mathbf{x}^{(0)}$ to each worker. Then, each worker sends a vector of dimension $N$,  $\mathbf{r}^{(n)}(\xi_l)+\mathbf{s}^{(n)}(\xi_l)$ to the master node. Therefore, communication cost per processing node is $\beta_1+2N\beta_2$.

     \emph{Storage cost}: Each worker stores $p_{\mathbf{Q}_2}(\xi_l)$,  $p_{\mathbf{I}}(\xi_l^{m^{i-1}})$,  $p_{\mathbf{A}_1}(\xi_l^{m^{i-1}})$ and $p_{\mathbf{A}_2}(\xi_l^{m^{i-1}})$ for $1\leq i\leq { \frac{{n}} {2}}$. Exploiting the sparsity of the identity matrix, the storage cost  corresponding to $p_{\mathbf{I}}(\xi_l^{m^{i-1}})$ is $O(N)$.

      Therefore, overall storage cost is $2(\frac{n}{2}\frac{N^2}{m})+\frac{N^2}{m}+N=O((n+1)\frac{N^2}{m})=O((n+1){N^2}/{(\frac{P+1}{2})}^{\frac{2}{n}})$.

     \emph{Computational complexity}: Computation cost of $i$-th worker can be summarized as follows:
     \begin{itemize}
     \item Computation of $\mathbf{r}^{(n)}$ which involves $n$ matrix-vector multiplications $n\frac{N^2}{m}=nN^2/(\frac{P+1}{2})^{\frac{2}{n}}$.
     \item Computation of $\mathbf{s}^{(n)}$ involves: \begin{itemize}
     \item[1)] $n$ matrix-vector multiplications due to computing $\mathbf{w}^{(i)}$ at each iteration, equivalent to $n\frac{N^2}{m}$ operations.
         \item[2)] Because of sparsity of $p_{\mathbf{I}}(\xi)$ which has $N$ non-zero entries, computing $p_{\mathbf{I}}(\xi)\mathbf{{w}}^{(i)}$ in odd iterations requires $\frac{n}{2}N$ operations,
     \end{itemize}
     \end{itemize}
     \vspace{-1em}
     As a consequence, the overall computation complexity is 
     $O({n}\frac{N^2}{m})=O(nN^2/{(\frac{P+1}{2})}^{\frac{2}{n}})$.

\emph{Pre-processing cost}: Preprocessing cost is due to computing evaluations of  $p_{\mathbf{A}_i}(\xi), p_{\mathbf{Q}_2}(\xi),p_{\mathbf{I}}(\xi)$ which is $O(N^2)$. Therefore, overall pre-processing complexity is $O(nPN^2)$.

\emph{Post-processing cost}: After $n$ iterations we need to recover a vector, $\mathbf{A}^n\mathbf{x}^{(0)}+\mathbf{A}^{n-1}\mathbf{Qy}+\dots+\mathbf{Qy}$, with $N$ elements. Therefore, post-processing requires interpolating a $P=2m^{\frac{n}{2}} -1$ polynomials of degree $P-1$. As a consequence of Corollary 1 in main draft, complexity per vector element is $O(P \log^2 P\log(\log P))$ and overall  complexity is $O(NP \log^2 P\log(\log P))$.
\section{Table of comparison}\label{e}
Here we report the performance of competing schemes from Table 1 of \cite{lee2017distributed}. Note that in this table for algorithms DISCO, DANE, COCOA$^+$, AccelGrad, DSVRG and DASVRG. {Note that we do not report the coefficient of $\beta_2$. We refer the reader to the cited papers for more detail.}

\begin{table*}[h]
\centering
\caption{Comparison of different parallelizing schemes with different measures. In this table, $Q=O(\frac{N^2}{P}\sqrt{\kappa}\log{\frac{1}{\epsilon}})$ \cite{lee2017distributed} with $n=\log_{\frac{1}{|\sigma_1|}}\frac{N\underset{1\leq i\leq N}{\max}|\alpha_i|}{\epsilon}$ and $n^*=n\log{\frac{1}{|\sigma_1|}}-\log (N\max|\alpha_i|)$.}
\resizebox{1.1\linewidth}{!}{
\begin{tabular}{c|cccc}
  \hline
  Strategy & Computational cost & Pre/Post-processing cost & Storage cost & Communication cost\\
  \hline
  \\
  {DISCO} & {$Qn^*(1+\frac{P^{0.25}\sqrt{\kappa}}{N^{0.25}})$} & $O(N)\:,\:O(nN)$ & $O(\frac{N^2}{P})$ & {$n^*\beta_1(1+\frac{P^{0.25}\sqrt{\kappa}}{N^{0.25}}))$} \\
\hline
  \\
   {DANE} &  {$ Qn^*(1+\frac{P\kappa^2}{N})$} & $O(N)\:,\:O(nN)$ & $O(\frac{N^2}{P})$ & {$n^*\beta_1(1+\frac{P\kappa^2}{N})$} \\
  \hline
    \\
    {COCOA$^+$} & {$O(n^*\frac{N^2}{P}+n^*N\sqrt{\kappa\frac{N}{P}})$} & $O(N)\:,\:O(nN)$ & $O(\frac{N^2}{P})$ & {$n^*\beta_1\kappa$} \\
      \hline
        \\
    {$Accel Grad$} &  {$O(n^*\frac{N^2}{P}\sqrt{\kappa_f})$} & $O(N)\:,\:O(nN)$ & $O(\frac{N^2}{P})$ &  {$n^*\beta_1\sqrt{\kappa_f}$}
    \\
            \hline
            \\
{DSVRG} & {$O(n^*\frac{N^2}{P}+n^*N\kappa)$} & $O(N)\:,\:O(nN)$ & $O(\frac{N^2}{P})$ & {$n^*\beta_1(1+\frac{P\kappa}{N})$}
\\
            \hline
            \\
{DASVRG} & {$O(n^*\frac{N^2}{P}+n^*N\sqrt{\frac{N}{P}\kappa})$} & $O(N)\:,\:O(nN)$ & $O(\frac{N^2}{P})$ & {$n^*\beta_1(1+\sqrt{\frac{P\kappa}{N}})\log(1+\frac{P\kappa}{N})$}
            \\
                        \hline
                        \\
 {Karakus \cite{karakus2017straggler}} & {$O(\frac{nN^2}{K}),K\leq P$} &  {$O(\frac{PN^3}{K}),O(nN)$} & {$O(\frac{N^2}{K})$}& {$\beta_1n+\beta_2n(2N)$}
                                    \\
                                                \hline
                                                \\
 BaselineParallel & $O(\frac{nN^2}{P})$ & $O(N)\:,\:O(nN)$ &$O(\frac{N^2}{P})$ & $\beta_1(n)+\beta_2(n\times\frac{1+P}{P}\times N)$ \\
 \hline
 \\
  PolyLin & $O({nN^2}/{{(\frac{K+1}{2})}^{\frac{2}{n}}}),K\leq P$ & $O(nPN^2) \:,\: O( NK\log^2K\log(\log K))$ & $O((n+1){N^2}/{(\frac{K+1}{2})}^{\frac{2}{n}})$ & $\beta_1+\beta_2(2N)$
  \\
  \hline
  \\
  MRPolyLin & $O({nN^2}/{{(\frac{K+1}{2})}^{\frac{2\ell}{n}}}),K\leq P$ & $O(\frac{n}{\ell}PN^2) \:,\:{O}(\ell NK\log^2K\log(\log K))$ & $O((\frac{n+\ell}{\ell}){N^2}/{(\frac{K+1}{2})}^{\frac{2\ell}{n}})$ & $\ell\beta_1+\beta_2(2\ell N)$
   \\
    \hline
\end{tabular}}
\vspace{-1.2em}
\end{table*}
\end{appendices}

\end{document}